\documentclass[a4paper]{article}
\usepackage{amsfonts,amssymb,amsmath,amsthm,cite}
\usepackage{ucs} 
\usepackage[utf8x]{inputenc}
\usepackage{graphicx}
\usepackage[english]{babel}
\usepackage{slashed}
\usepackage{textcomp}
\usepackage{bm}
\usepackage{titlesec}

\evensidemargin=\oddsidemargin
\newtheorem{lemma}{Lemma}

\begin{document}
\begin{center}
	{\Large\textbf{On Local Heat Kernel}}
	\vspace{0.5cm}
	
	{\large Aleksandr~V.~Ivanov}
	
	\vspace{0.5cm}
	
	{\it St. Petersburg Department of Steklov Mathematical Institute of
		Russian Academy of Sciences,}\\{\it 27 Fontanka, St. Petersburg 191023, Russia}\\
	{\it Leonhard Euler International Mathematical Institute, 10 Pesochnaya nab.,}\\
	{\it St. Petersburg 197022, Russia}\\
	{\it E-mail: regul1@mail.ru}
\end{center}
%\vskip 10mm
\date{\vskip 20mm}
%\begin{flushright}
%	\Large{{\calligra Dedicated to our parents}}
%\end{flushright}
%\vskip 10mm
\begin{abstract}
The paper is devoted to a local heat kernel, which is a special part of the standard heat kernel. Locality means that all considerations are produced in an open convex set of a smooth Riemannian manifold. We study such properties and concepts as uniqueness, a symmetry of the Seeley--DeWitt coefficients, extension on the whole manifold, a family of special functions, and the late-time asymptotics with the use of the path integral formulation.
\end{abstract}
\normalsize
\tableofcontents

\section{Introduction}
Heat kernels play important roles in the modern theoretical physics and mathematics, see \cite{10,7,103,fuvas}. They appear in very different areas, for example, in the Atiyah--Patodi--Singer theorem \cite{1-2,2-2} or in the renormalization of quantum field models \cite{12,13,Ivanov-Kharuk-2020,Iv-Kh-2022}. A wider range of their applications can be found in \cite{vas1,33}. It is very well known that their explicit formulae are possible only in some special cases, so this fact results in the necessity of working with asymptotic expansions, coefficients of which can be calculated recurrently. One of the expansions is a base element of the proper time method \cite{1,108}, which allows us to investigate spectral functions for elliptic operators, for instance, Laplace-type or Dirac ones.
Unfortunately, in papers devoted to applications the boundary between the standard heat kernel and the main part of its asymptotic expansion is erased, because their difference does not give a contribution in the final answers. However, they are unequal objects, and it would be helpful to keep in the mind their features.

In this paper, we study a local heat kernel \cite{iv-kh-2022} of the Laplace-type operator (\ref{op}), which is a special part of the asymptotic expansion of the standard heat kernel for small values of the proper time. In our case the locality means that we work in some smooth open convex (normal neighborhood of each of its points) set $U$ of a smooth Riemannian manifold $\mathcal{M}$. In particular, the last condition means that information from $\mathcal{M}\setminus U$ and any boundary conditions do not affect the local heat kernel. The main aim of the work is to discuss such properties as uniqueness, symmetry of the coefficients, extension on $\mathcal{M}$, family of special functions, and the late-time asymptotics.

The paper has the following structure. In Section \ref{defdef}, we introduce the local heat kernel and describe his difference from the standard heat kernel. Then, in Section \ref{propprop}, we study a number of properties mentioned above, prove lemmas, and derive some additional formulae. In the conclusion, we discuss the results and formulate some interesting tasks for further work.

\section{Definition}
\label{defdef}

Let $\mathcal{M}$ be a $d$-dimensional smooth Riemannian manifold, and $\mathcal{A}$ is its atlas. Here, the smoothness means that the manifold is Hausdorff and paracompact as well. For clarity, we give definitions for local objects in a coordinate chart $(U_\alpha,\phi_\alpha)\in\mathcal{A}$.
For example, $g^{\mu\nu}(x)$ with $\mu,\nu\in\{1,\ldots,d\}$ is the metric tensor at $\phi_\alpha^{-1}(x)=p\in U_\alpha\subset\mathcal{M}$. Of course, the last matrix-valued operator is symmetric and real-valued. Then, we introduce a Hermitian smooth vector bundle $\mathcal{H}$ over $\mathcal{M}$ and the corresponding components $B_\mu(x)$ of the Yang--Mills connection 1-form. After that we can 
define a Laplace operator in local coordinates
\begin{equation}
\label{op}
A(x)=-g^{-1/2}(x)D_{x^\mu} g^{1/2}(x)g^{\mu\nu}(x)D_{x^\nu}-v(x).
\end{equation}
Here $D_{x^\mu}=\partial_{x^\mu}+B_\mu(x)$ is the covariant derivative, $v(x)$ is a smooth matrix-valued Hermitian potential, and $g(x)$ is the metric tensor determinant. We assume that our Laplace operator is equipped with appropriate boundary conditions on $\mathcal{M}$, such that the spectral problem is well-posed and the operator is symmetric. In addition, we assume that the operator coefficients have non-zero convergence radius of the Taylor series at each point of $\mathcal{M}$. Moreover we require that this radius is always more than a fixed positive constant.

After that we can formulate a problem for a standard heat kernel $\hat{K}(p,q;\tau)$. Let $p,q\in\mathcal{M}$, $p\in U_\alpha$ and $q\in U_\beta$, $(U_\alpha,\phi_\alpha)$ and $(U_\beta,\phi_\beta)$ are from the atlas $\mathcal{A}$, then we have
\begin{equation}
	\label{tepl}
\begin{cases}
\big(\partial_\tau+A\big(\phi_\alpha(p)\big)\big)\hat{K}(p,q;\tau)=0\,\,\,\mbox{for all}\,\,\,\tau>0;
		\\
\hat{K}(p,q;0)=g^{-1/2}\big(\phi_\alpha(p)\big)\delta(p-q)\,\,\,\mbox{and boundary conditions}.
\end{cases}
\end{equation}
As we see from the problem statement, the standard heat kernel is a global object and depends on the boundary conditions. Unlike him, the main object of the paper, a local heat kernel, does not inherit the last two properties.

For further consideration, let us introduce an open convex set $U\subset\mathcal{M}$, which is a normal neighborhood of each of its points by definition. Such set contains a unique geodesic segment of any two points $p,q\in U$. According to Proposition 7 from Chapter 5 of \cite{neil}, each point of $\mathcal{M}$ has a convex neighborhood, so we obtain the existence of $U$. For clarity, we assume that $U\subset U_\alpha$, such that $\phi_\alpha(p)=x$ and $\phi_\alpha(q)=y$. Moreover, we require that the set is so small, that all the coefficients would be decomposable in the covariant Taylor series.
Now we are ready to introduce a local heat kernel $K(x,y;\tau)$, which is the solution of the following problem
\begin{equation}
	\label{tepl1}
	\begin{cases}
		\big(\partial_\tau+A(x)\big)K(x,y;\tau)=0\,\,\,\mbox{for all}\,\,\,\tau>0;
		\\
		K(x,y;0)=g^{-1/2}(x)\delta(x-y)\,\,\,\mbox{and asymptotic behaviour at}
		\,\,\,\tau\to+0.
	\end{cases}
\end{equation}
Here, the last condition can be presented as the following series for $\tau\to+0$
\begin{equation}\label{K_0}
	K(x,y;\tau)=\frac{\Delta^{1/2}(x,y)}{(4\pi\tau)^{d/2}}e^{-\sigma(x,y)/2\tau}\sum_{k=0}^{+\infty}\tau^ka_k(x,y),
\end{equation}
where $a_k(x,y)$, $k\geqslant0$, are the Seeley--DeWitt (or Hadamard, Minakshisundaram \cite{111}, and Gilkey \cite{8}) coefficients, see \cite{110,10}. Then, $\sigma(x,y)$ is Synge’s world function \cite{104,1040}, and $\Delta(x,y)$ is Van-Vleck--Morette determinant \cite{105}, which is defined by the formula
\begin{equation}\label{op3}
	\Delta(x,y)=\left(g(x)g(y)\right)^{-1/2}\det\left(-\frac{\partial^2\sigma(x,y)}{\partial x^\mu\partial y^\nu}\right).
\end{equation}

As follows from the definitions, the kernels can be compared only in $U$, because the second one exists only there. Actually, the local heat kernel is included into the standard heat kernel as a part, which leads to the $\delta$-function in the limit $\tau\to+0$. Besides that the standard kernel contains an additional part, which makes the boundary conditions satisfied. So, for $p,q\in U$, mentioned above, we can write
\begin{equation}
\lim_{\tau\to+0}\Big(\hat{K}(\phi_\alpha^{-1}(x),\phi_\alpha^{-1}(y);\tau)-K(x,y;\tau)\Big)=0.
\end{equation}
Let us note that in some special cases both kernels can be equal to each other, for example, for $A(x)=-\partial_{x_\mu}\partial_{x^\mu}$ on $\mathbb{R}^d$. In this case we can take only one convex set $U=\mathbb{R}^d$.

\section{Properties}
\label{propprop}
Most of the functions, described in this section have arguments $x$ and $y$, such that $\phi_\alpha^{-1}(x),\phi_\alpha^{-1}(y)\in U$. Therefore, it is convenient to define simplified notations 
\begin{equation}\label{op1}
K(\tau)=K(x,y;\tau),\,\,\Delta=\Delta(x,y),\,\,
a_k=a_k(x,y),\,\, D_\mu=D_{x^\mu},\,\, A=A(x),
\end{equation}
\begin{equation}\label{op2}
\sigma=\sigma(x,y),\,\,
\sigma_\mu=\partial_{x^\mu}\sigma,\,\,\sigma^\mu=\partial_{x_\mu}\sigma.
\end{equation}

\subsection{Uniqueness and symmetry}
The asymptotic behaviour of the local heat kernel at $\tau\to+0$ is a quite strong restriction, which makes the solution of problem (\ref{tepl1}) to be unique in a sense. Indeed, after substitution of (\ref{K_0}) into the problem we obtain quite remarkable recurrence relations \cite{10}
\begin{equation}\label{rec}
	\sigma^\mu D_\mu a_0=0,\,\,\, a_0|_{y=x}=1,\,\,\,\mbox{and}\,\,\,
	(k+1+\sigma^\mu D_\mu)a_{k+1}=-\Delta^{-1/2}A\Delta^{1/2}a_k,\,\,k\geqslant0.
\end{equation}

As it is known, the first equation for $a_0$ leads to a path-ordered exponential, which is smooth by construction. Then, the equation for $a_1$ contains non-zero right hand side, which is smooth as well. Hence, the answer for $a_1$ can be presented as a sum of a smooth part and a kernel of the operator $(1+\sigma^\mu D_\mu)$ in the form
\begin{equation}\label{op4}
a_1=\mbox{(smooth part)}+\alpha a_0\sigma^{-1/2},
\end{equation}
where $\alpha=\alpha(x,y)$ is a kernel of $\sigma^\mu\partial_\mu$. In the last calculation we have used $\sigma_\mu\sigma^\mu=2\sigma$, see formula (2.20) in \cite{1040}. Keeping only the smooth part, we move on to the next coefficient. Continuing this procedure leads in the k-th order to the form
\begin{equation}\label{op4}
	a_k=\mbox{(smooth part)}+\alpha a_0\sigma^{-k/2},
\end{equation}
and so on. After this consideration we can formulate the result.
\begin{lemma}\label{lll1}
Under the conditions described above, a solution of problem (\ref{tepl1}) in the form of asymptotic series (\ref{K_0}) with smooth coefficients is unique.
\end{lemma}

There are a lot of ways to compute the smooth parts, including simplified particular cases. Unfortunately, this is out of the scope of our paper, so we only list a number of links on some general methods \cite{33,f1,f2,f3,f4,f5,f6,f7,f8}. In addition, we note that an explicit closed formula exists only on the diagonal, $y=x$, thanks to which the Seeley--DeWitt coefficient $a_k|_{y=x}$ can be represented as a finite nonlinear combination of the potential, components of the connection 1-form, metric, and their covariant derivatives.

In the rest of this subsection we want to discuss one unevident property of the Seeley--DeWitt coefficients. It says that the Hermitian conjugation leads to a permutation of the arguments. Actually, there are two ways to prove this relation. The first one \cite{sym1,sym2}, hard way, is related to the analysis of internal functional properties of the coefficients and the recurrence relations from (\ref{rec}), while the second one \cite{sym3} is connected with the analysis of asymptotics for the standard heat kernel at $\tau\to+0$. Unfortunately, the last proof was done only for compact manifolds.

Actually, we can expand the second way on our case, because we work with the local heat kernel on the set $U$. Hence, we need to formulate a problem for the standard heat kernel on $U$, for example, in the following form
\begin{equation}
	\label{tepl2}
	\begin{cases}
		\big(\partial_\tau+A(x)\big)\hat{K}_{U}(x,y;\tau)=0\,\,\mbox{for all}\,\,\, x,y\in U,\,\,\mbox{and}\,\,\,\tau>0;
		\\
		\hat{K}_U(x,y;0)=g^{-1/2}(x)\delta(x-y)\,\,\mbox{and Dirichlet condition on}\,\,\partial\overline{U}.
	\end{cases}
\end{equation}
Then, according to the result from \cite{sym3}, we obtain the symmetry, mentioned above. Therefore, we obtain the symmetry for the local heat kernel, because the Seeley--Dewitt coefficients are unique. So, we can formulate the result.
\begin{lemma}\label{lll3}
Under the conditions described above, we have $\big(a_k(x,y)\big)^\dag=a_k(y,x)$ for $k\geqslant0$, and, in particular,
\begin{equation}\label{op8}
\big(K(x,y;\tau)\big)^\dag=K(y,x;\tau).
\end{equation}
\end{lemma}

\subsection{Continuation}
In the next stage we want to discuss global properties of the local heat kernel. We know, that it exists and unique in some small convex neighborhood. But what can we say about extension on $\mathcal{M}$? Actually, we have two quite strong restrictions in the general case. The first one is connected with the ability to decompose the coefficients of the operator $A$ into the covariant Taylor series. The second one is related to the need to have only one unique geodesic in each neighborhood for two selected points, because otherwise Synge's world function loses its smoothness.

We offer the following procedure. Let us take an open covering $\mathcal{C}$ of the manifold $\mathcal{M}$, in which any open set $V\in\mathcal{C}$ is included in a set $U_\alpha$ from the atlas $(U_\alpha,\phi_\alpha)\in\mathcal{A}$, and it is such small that the coefficients of the operator $A$ are decomposable in the covariant Taylor series. Further, let us define a convex covering $\mathcal{R}$ of the manifold as a covering of $\mathcal{M}$ by convex open sets, such that
\begin{equation}\label{op6}
\mbox{if}\,\, V_1,V_2\in\mathcal{R}\,\,\mbox{and}\,\,
V_1\cap V_2\neq\varnothing,\,\,\mbox{then}\,\, V_1\cap V_2\,\,\mbox{is convex}.
\end{equation}
After that, using Lemma 10 from Chapter 5 of \cite{neil}, we require that the covering $\mathcal{R}$ has one more property: each element of $\mathcal{R}$ is contained in some element of $\mathcal{C}$.

This means that each element $V$ of $\mathcal{R}$ is an open convex set by construction and it satisfies two restrictions, mentioned above. Hence, to each element of $\mathcal{R}$ we can apply Lemma \ref{lll1}. In addition, the lemma is applicable to $V_1\cap V_2$, if $V_1,V_2\in\mathcal{R}$ and $V_1\cap V_2\neq\varnothing$. This means that we can transfer the local heat kernel from $V_1$ to $V_2$, see Figure \ref{cont}. 
Indeed, let us study this process in details.
We start from $q_1,p_1\in V_1$ and finish at $q_2,p_2\in V_2$. 
Let the corresponding charts be $(U_\alpha,\phi_\alpha)$ and $(U_\beta,\phi_\beta)$, respectively. Then we introduce two paths $\gamma_i:[0,1]\to V_1\cap V_2$, $i=1,2$, such that
\begin{equation}\label{op11}
\gamma_1(0)=q_1,\,\,\gamma_2(0)=p_1,\,\,\gamma_1(1)=q_2,\,\,\gamma_2(1)=p_2,
\end{equation}
\begin{equation}\label{op7}
\mbox{and there is such}\,\,
s\in [0,1]\,\,\mbox{that}\,\,
\gamma_i:[0,s]\to V_1\,\,\mbox{and}\,\,\gamma_i:[s,1]\to V_2.
\end{equation}
Hence, our local heat kernel moves from $V_1$ to $V_2$, when the parameter goes from $0$ to $1$. Note, that at the point $s$ we make the change of the local coordinates.

Moreover, we can transfer the local heat kernel from any $V_1\in\mathcal{R}$ to any $V_2\in\mathcal{R}$, because we can choose a number of elements from $\mathcal{R}$, which connect $V_1$ and $V_2$. Besides that, such movements do not depend on a path, because in each neighborhood the local heat kernel is unique.

Now we can formulate the following result.
\begin{lemma}\label{lll2}
Under the conditions described above, the local heat kernel is unique on
\begin{equation}\label{op5}
\mathfrak{R}=\bigcup_{V\in\mathcal{R}}V\times V.
\end{equation}
\end{lemma}
\begin{figure}[h]
	\centerline{\includegraphics[width=0.3\linewidth]{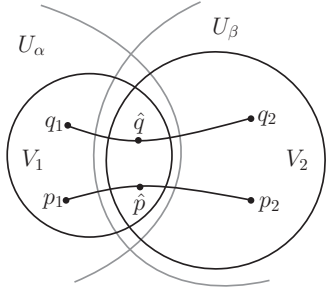}}
	\caption{Arguments transfer of the local heat kernel from $q_1,p_1\in V_1\subset U_\alpha$ to $q_2,p_2\in V_2\subset U_\beta$ along the curves $\gamma_i:[0,1]\to V_1\cup V_2$, $i=1,2$. Here, $\hat{q}=\gamma_1(s)$ and $\hat{p}=\gamma_2(s)$ for some $s\in[0,1]$.}
	\label{cont}
\end{figure}

Similar discussions were written in the recent paper \cite{moret} devoted to Synge's world function. We have expanded them on the local heat kernel. Let us additionally note, that $\mathfrak{R}$ is actually a neighborhood of the diagonal of $\mathcal{M}\times\mathcal{M}$. It is not unique, this means that we can construct another set $\mathcal{R}^{\prime}$, satisfying properties mentioned above. In this case we can get the third set $\mathcal{R}^{\prime\prime}$, such that
\begin{equation}\label{op9}
	\mbox{any}\,\,V^{\prime\prime}\in\mathcal{R}^{\prime\prime}\,\,
	\mbox{lies in}\,\,V\cap V^\prime\,\,\mbox{for some}\,\,
	V\in\mathcal{R}\,\,\mbox{and}\,\,V^\prime\in\mathcal{R}^\prime.
\end{equation}
Hence, using the uniqueness, we obtain that the local heat kernel has the same values on the overlapping sets from $\mathcal{R}$ and $\mathcal{R}^\prime$. Unfortunately, after such procedure we decrease the maximum distance between $x$ and $y$ in each set. Some comments on the “maximum” covering can be found in the conclusion.

As a closing of the subsection, we would like to mention two interesting situations. The first one is a case of compact manifold, which leads to $\mathcal{R}$ with a finite number of elements. The second one is $\mathcal{M}=\mathbb{R}^d$ with $g^{\mu\nu}(x)=\delta^{\mu\nu}$ for all points. If the operator coefficients have an infinitely large convergence radius, then we can get only one set $\mathcal{R}=\{\mathcal{M}\}$.

\subsection{Special functions}

In this subsection, we describe a family of special functions applicable to decomposition of the local heat kernel. We follow notations from the recent paper \cite{iv-kh-2022}, in which the functions were introduced. Let us note that taking into account the results of Lemma \ref{lll2}, we assume that an open convex set $U$ is from $\mathcal{R}$. Moreover, we omit the arguments $x,y\in U$, if this does not cause any confusion.

Let us introduce a set of functions
\begin{equation}\label{fun1}
\Psi_k=\Delta^{1/2}\sum_{n=0}^{+\infty}
\frac{(-\sigma/2)^{n-k}a_n}{\Gamma(n-k+1)},
\end{equation}
where $k\in\mathbb{Z}$. Our definition corresponds to $\Psi_k^\omega$ with $\omega=\sigma(x,y)$ from formula (25) in \cite{iv-kh-2022}. These functions satisfy quite remarkable relation, see Lemma 1 in \cite{iv-kh-2022},
\begin{equation}\label{fun2}
A\Psi_k=(d/2-1-k)\Psi_{k+1}\,\,\,
\mbox{for}\,\,k\in\mathbb{Z},
\end{equation}
which can be used in different proofs and derivations.

Using the decomposition of the exponential from (\ref{K_0}) in powers of $\tau^{-k}$ and reversing the order of the sums, we obtain the following result.
\begin{lemma}\label{lll4}
Let $x,y\in U\in\mathcal{R}$, constructed above. Then the local heat kernel (\ref{K_0}) has the following representation in terms of functions (\ref{fun1})
\begin{equation}\label{fun3}
K(\tau)=\frac{1}{(4\pi\tau)^{d/2}}\sum_{k\in\mathbb{Z}}
\tau^{k}\Psi_k.
\end{equation}
\end{lemma}

For further discussion, we need to introduce two auxiliary objects for an even-dimensional case
\begin{equation}\label{fun4}
K_-(\tau)=\frac{1}{(4\pi)^{d/2}}\sum_{k=0}^{+\infty}
\frac{\Psi_{d/2-1-k}}{\tau^{1+k}},\,\,\,
K_+(\tau)=\frac{1}{(4\pi)^{d/2}}\sum_{k=0}^{+\infty}
\tau^{k}\Psi_{d/2+k}.
\end{equation}
It is quite interesting that both functions satisfy heat kernel equation $(\partial_\tau+A)K_\pm(\tau)=0$. This is possible due to the presence of the local zero modes $A\Psi_{d/2-1}=0$ in the even-dimensional case.

\begin{lemma}\label{lll5}
Let $x,y\in U\in\mathcal{R}$, constructed above, $\tau>0$, and $s\in\mathbb{R}$, such that $\tau+s>0$. Then the local heat kernel (\ref{K_0}) has the following structure
\begin{equation}\label{fun5}
K(\tau+s)=e^{-sA}K(\tau)=e^{s\partial_\tau}K(\tau).
\end{equation}
If the dimension $d$ is even, then $K(\tau)=K_-(\tau)+K_+(\tau)$, and
\begin{equation}\label{fun6}
K_\pm(\tau+s)=e^{-sA}K_\pm(\tau)=e^{s\partial_\tau}K_\pm(\tau).
\end{equation}
\end{lemma}
\begin{proof} Let us start from the first relation from (\ref{fun5}) under the assumption $|s|<\tau$. It follows from the straight application of the exponential $\exp(-sA)$ to the local heat kernel with the usage of relation (\ref{fun2}) and the following identity
\begin{equation}\label{fun9}
\frac{1}{\tau^{k}}\sum_{n=0}^{+\infty}\frac{\Gamma(k+n)(-s/\tau)^n}{\Gamma(k)\Gamma(n+1)}=\frac{1}{(\tau+s)^k}.
\end{equation}
The last series converges, because $|s/\tau|<1$. For $-k\in\mathbb{N}$ we obtain the standard binomial expansion.
Further, the case $s>\tau$ follows from the previous one, because we can take such $N\in\mathbb{N}$, that $s/N<\tau$. Hence, we can rewrite the exponential
as the product
\begin{equation}\label{fun10}
e^{-sA}=\prod_{i=1}^Ne^{-(s/N)A}
\end{equation}
and apply them $N$ times.

Then, the second equality in (\ref{fun5}) follows from the replacement of $A$ by $-\partial_\tau$ with the use of the heat equation from (\ref{tepl1}). Formulae from (\ref{fun6}) can be obtained in the same manner.
\end{proof}

Actually, using the last two lemmas, we can rewrite the local heat kernel in a more elegant form
\begin{equation}\label{fun7}
K(\tau)=e^{-\tau A}\delta_A\,\,\,
\mbox{with}\,\,\,\delta_A(x,y)=\lim_{\epsilon\to+0}K(x,y;\epsilon),
\end{equation}
where we have introduced a regularization of the $\delta$-function. It means that we can put the exponential $\exp(-\tau A)$ under the limit and apply Lemma \ref{lll5}.

\subsection{Late-time asymptotics}
The last subsection is devoted to the discussion of the late-time asymptotics $\tau\to+\infty$, one more non-trivial characteristic of the local heat kernel. In comparison with the short-time asymptotics $\tau\to+0$, which actually is embedded into definition (\ref{tepl1}), a form of the late-time behavior is not evident and depends on the local properties of the potentials. 

Indeed, as an example, we can consider $\mathbb{R}^d$ with $g^{\mu\nu}(x)=\delta^{\mu\nu}$ for all values of the argument. Let us construct the local heat kernel for the operator $-\partial_{x_\mu}\partial_{x^\mu}-v(x)$ in two disjoint open convex sets $V_1$ and $V_2$, in which the potential $v(x)$ is equal to the constants $c_1,c_2\in\mathbb{R}$, respectively. We assume that $c_1\neq c_2$. Hence, we obtain significantly different behavior of the local heat kernels
\begin{equation}\label{la5}
K(x,y;\tau)=\frac{e^{-|x-y|^2/4\tau}}{(4\pi\tau)^{d/2}}e^{c_i\tau}\,\,\,
\mbox{for}\,\,\,x_i,y_i\in V_i,\,\,\,i=1,2.
\end{equation}
Their asymptotics differ by an exponential factor. This simple example shows that we cannot introduce a unique type of the late-time behavior, as this was made in the short-time case (\ref{K_0}).

As a rule, the late-time asymptotics is studied in the context of the standard heat kernel, because in the case we have methods of the spectral theory of differential operators, which are quite powerful tools of the modern mathematical physics. It is useful to remember two popular ways for the mentioned problem statement: the first one \cite{103,lala8} gives estimates in terms of the lowest eigenvalue on compact manifolds, while the second one \cite{lala1,lala2} suggests to use ansatz of a special type.

Unfortunately, standard spectral methods do not work in the case of the local heat kernel, because we ignore boundary conditions and, as a result, we lose a well-posed spectral problem. However, we have a different tool, path integral, which leads to some results. Now, using notation and definitions from \cite{lala5,lala4,lala3,lala7,lala6,33}, we formulate some interesting relations.

Let $V\in\mathcal{R}$ and $(U_\alpha,\phi_\alpha)\in\mathcal{A}$, such that $V\in U_\alpha$. In this case we obtain a star-shaped set $\phi_\alpha(V)=\tilde{V}\subset\mathbb{R}^d$. As it was noted in the introduction, we assume that $g^{\mu\nu}(x)$, $B_\mu(x)$, and $v(x)$, are decomposable in the convergent Taylor series on $\tilde{V}$. Let us by hands expand the definitions on $\mathbb{R}^d$ as follows
\begin{equation}\label{la6}
g^{\mu\nu}(x),B_\mu(x),v(x)\,\,\mbox{on $\tilde{V}$}\to
\tilde{g}^{\mu\nu}(x),\tilde{B}_\mu(x),\tilde{v}(x)\,\,\mbox{on $\mathbb{R}^d$},
\end{equation}
such that both types of objects are equal to each other on $\tilde{V}$. 
Moreover, we require some additional stronger conditions: $\tilde{g}^{\mu\nu}(x)=\delta^{\mu\nu}$, and all the potentials are decomposable in the convergent Taylor series. Then, according to Theorems 4 and 5 from \cite{33}, we can write the following representation for the local heat kernel
\begin{equation}\label{la1}
K(x,y;\tau)=\frac{1}{(4\pi)^{d/2}}
\int_{\mathcal{W}_{y,x}}\mathcal{D}_\tau u\,e^{-S_\tau[u]}
P_t\exp\bigg[\int_0^\tau dt\,M_1(u(t))\bigg]
\,\,\,\mbox{with}\,\,\,
\int_{\mathcal{W}_{0,0}}\mathcal{D}_\tau u\,e^{-S_\tau[u]}=\tau^{-d/2},
\end{equation}
where we have used
\begin{equation}\label{la2}
S_\tau[u]=\frac{1}{4}\int_0^\tau dt\,\dot{u}_\mu(t)\dot{u}^\mu(t),\,\,\,
M_s(u(t))=-\dot{u}^\mu(t)\tilde{B}_\mu(u(t))+s\tilde{v}(u(t)),
\end{equation}
$\mathcal{W}_{y,x}$ is a set of continuous paths in $\mathbb{R}^d$ with the beginning at $y$ and end at the point $x$, and the dot $\dot{u}$ means the derivative $du(t)/dt$. We draw the attention that in general case $M(\cdot)$ is a matrix-valued function, because of what we have obtained the time-ordered $P_t$ exponential instead of the ordinary one.

Now we transform representation (\ref{la1}) to emphasize its dependence on the parameter $\tau$. Firstly, making the following change $u(t)\to u(t)+y+t(x-y)/\tau$, we move on from the set $\mathcal{W}_{y,x}$ of continuous paths to the set $\mathcal{W}_{0,0}$ of continuous loops. Then, we make the chain of changes 
\begin{equation}\label{la4}
	t\to t\tau,\,\,\,
u(\tau t)\to\sqrt{\tau}u(t),\,\,\,\mbox{and}\,\,\,\mathcal{D}_\tau u\to\mathcal{D}_1 u,
\end{equation}
after which we obtain one more representation
\begin{equation}\label{la3}
K(x,y;\tau)=
\frac{e^{-|x-y|^2/4\tau}}{(4\pi\tau)^{d/2}}
\int_{\mathcal{W}_{0,0}}\mathcal{D}_1 u\,e^{-S_1[u]}
P_t\exp\bigg[\int_0^1 dt\,M_\tau\big(\sqrt{\tau}u+y+t(x-y)\big)\bigg],
\end{equation}
This formula is convenient, because it is factorized. Indeed, its first part coincides with the first one from the local heat kernel (\ref{K_0}), while the second factor, the path integral, represents the sum of the Seeley--DeWitt coefficients. At the same time, all the dependence on the parameter $\tau$ is included into the time-ordered exponential. For a simple check, we can note that
\begin{equation}\label{la11}
P_t\exp\bigg[\int_0^1 dt\,M_\tau\big(\sqrt{\tau}u+y+t(x-y)\big)\bigg]\bigg|_{\tau=0}=a_0(x,y).
\end{equation}

Further, we make one more change in the following form $u(t)\to u(t)-(y-z)/\sqrt{\tau}-t(x-y)/\sqrt{\tau}$, where $z\in\mathbb{R}^d$ is an auxiliary point. Hence, we obtain
\begin{equation}\label{la33}
	K(x,y;\tau)=
	\frac{1}{(4\pi\tau)^{d/2}}
	\int_{\mathcal{W}_{(y-z)/\sqrt{\tau},(x-z)/\sqrt{\tau}}}\mathcal{D}_1 u\,e^{-S_1[u]}
	P_t\exp\bigg[\int_0^1 dt\,M_\tau\big(\sqrt{\tau}u+z\big)\bigg],
\end{equation} 
on the basis of which we can formulate one more result.
\begin{lemma}\label{lll6}
Let us understand the local heat kernel (\ref{K_0}) as the functional of the connection components $\tilde{B}_\mu(\cdot)$ and the potential $\tilde{v}(\cdot)$. Then, under the conditions described above, we can write
\begin{equation}\label{la34}
K(x,y;\tau)\big[\tilde{B}_\mu(\cdot),\tilde{v}(\cdot)\big]=\tau^{-d/2}
K\big((x-z)/\sqrt{\tau},(y-z)/\sqrt{\tau};1\big)\big[\sqrt{\tau}\tilde{B}_\mu\big(\sqrt{\tau}(\cdot)+z\big),
\tau\tilde{v}\big(\sqrt{\tau}(\cdot)+z\big)\big],
\end{equation}
and, in particular, when $z=y=x$, we get
\begin{equation}\label{la35}
K(x,x;\tau)\big[\tilde{B}_\mu(\cdot),\tilde{v}(\cdot)\big]=\tau^{-d/2}
K(0,0;1)\big[\sqrt{\tau}\tilde{B}_\mu\big(\sqrt{\tau}(\cdot)+x\big),
\tau\tilde{v}\big(\sqrt{\tau}(\cdot)+x\big)\big].
\end{equation}
\end{lemma}
The last formula has a quite natural structure and can be obtained straight from the formula (\ref{K_0}) with the use of an appropriate scaling of the coordinates, the connection components, and the potential, 
in the Seeley--DeWitt coefficients. Indeed, using the structure of the coefficients on the diagonal, see \cite{33,f7}, and their locality and applying  the following chain of changes
\begin{equation}\label{la36}
	\tilde{B}_\mu(\cdot)\to\tilde{B}_\mu\big((\cdot)/\sqrt{\tau}\big)/\sqrt{\tau},\,\,\,
	\tilde{v}(\cdot)\to\tilde{v}\big((\cdot)/\sqrt{\tau}\big)/\tau,\,\,\,
	x\to\sqrt{\tau} x,
\end{equation}
we obtain the corresponding change of the coefficient.
\begin{equation}\label{la37}
a_k(x,x)\to a_k(x,x)/\tau^k\,\,\,\mbox{for}\,\,\,k\in\mathbb{N}\cup\{0\}.
\end{equation}

Let us formulate some examples. The first case is connected with quite popular illustrative situation, when the functions $\tilde{B}_\mu(x)=\xi_\mu$ and $\tilde{v}=c$ actually do not depend on the variables and, moreover, Abelian. In this case we get
\begin{equation}\label{la7}
\int_0^1 dt\,M_\tau\big(\sqrt{\tau}u+y+t(x-y)\big)=-(x-y)^\mu\xi_\mu+\tau c
\end{equation}
and
\begin{equation}\label{la8}
K(x,y;\tau)=
\frac{e^{-|x-y|^2/4\tau}}{(4\pi\tau)^{d/2}}e^{-(x-y)^\mu\xi_\mu+\tau c}\,\,\,
\mbox{for all}\,\,\,x,y\in\mathbb{R}^d\,\,\,\mbox{and}\,\,\,\tau>0.
\end{equation}

The second example allows us to make one simple estimate. Indeed, let the components of the connection 1-form are equal to zero $\tilde{B}_\mu=0$, while the potential $\tilde{v}(x)$ is scalar and satisfies the equality $v(x)<c$ for a constant $c$ and all values of the argument $x$. Then we obtain
\begin{equation}\label{la9}
P_t\exp\bigg[\int_0^1 dt\,M_\tau\big(\sqrt{\tau}u+y+t(x-y)\big)\bigg]<e^{\tau c}
\end{equation}
and
\begin{equation}\label{la10}
K(x,y;\tau)<\frac{e^{-|x-y|^2/4\tau}}{(4\pi\tau)^{d/2}}e^{\tau c}\,\,\,
\mbox{for all}\,\,\,x,y\in\mathbb{R}^d\,\,\,\mbox{and}\,\,\,\tau>0.
\end{equation}

\section{Conclusion}

In this paper, we have studied the local heat kernel (\ref{K_0}), which is actually the $\delta$-part of the standard heat kernel. We have clearly shown distinctive features of both objects and their connection. We have derived new relations and representation formulae for the local heat kernel, gave a number of demonstrative examples, and discussed such properties as globality, uniqueness, and behavior of asymptotics in different regions.

Additionally, it would be convenient to give one comment on the continuation of the local heat kernel. As we have shown, we can construct the convex covering, which allows us to transfer the object from one set to another one. But this covering is not unique. And the natural question arises: how can we construct the set with the largest neighborhoods to expand the definition of the local heat kernel $K(x,y;\tau)$ as far from the diagonal $x=y$ as possible? Of course, this new covering should be convex according to the definition formulated above.

Let us note that the local heat kernel in the fixed open convex set $U$ does not form the standard $C_0$-semigroup structure, because
\begin{equation}\label{fun11}
\int_{\overline{U}}d^dz\,\sqrt{g(z)}K(x,z;\tau_1)K(z,y;\tau_2)\neq K(x,y;\tau_1+\tau_2)
\end{equation}
in general case for $\tau_1,\tau_2>0$. This equality contains boundary terms, which follows after the integration by parts.

\vskip 5mm
\textbf{Acknowledgments.} Author is supported by the Ministry of Science and Higher Education of the Russian Federation, grant  075-15-2022-289, and by the Foundation for the Advancement of Theoretical
Physics and Mathematics “BASIS”, grant “Young Russian Mathematics”. Also, the author would like to thank N.V. Kharuk and D.V. Vassilevich for useful comments.

\end{document}